\documentclass[sigconf]{acmart}
\usepackage{colortbl}
\usepackage{threeparttable}
\usepackage{multirow}
\usepackage{graphicx,pifont}
\usepackage{xcolor}

\usepackage{amsmath,amssymb,amsfonts}
\usepackage{textcomp}
\usepackage{algorithm}
\usepackage{algorithmicx}
\usepackage{algpseudocode}
\renewcommand\footnotetextcopyrightpermission[1]{}

\usepackage{geometry}
\theoremstyle{definition} 
\newtheorem{definition}{Definition}[section] 
\theoremstyle{plain} 
\newtheorem{theorem}{Theorem}[section] 
\theoremstyle{remark} 

\usepackage{siunitx}
\usepackage{hyperref}
\hypersetup{
    colorlinks=true,
    linkcolor=blue,
    filecolor=magenta,      
    urlcolor=cyan,
    }
\newcommand{\rramdp}{\textit{$\mathsf{RRAM}$-$\mathsf{DP}$}}
\acmConference[ICCAD'2026]{International Conference on Computer-Aided Design}{November 8--12, 2026}{San Jose, CA, USA}

\begin{document}

\title[RRAM-DP: Device-Calibrated Differential Privacy for In-Memory Edge Learning]{
RRAM-DP: Device-Calibrated Differential Privacy for\\
In-Memory Edge Learning
}

\author{\fontsize{10}{8}\selectfont Kwunhang Wong$^{1,2,*}$, Jichang Yang$^{1,2,*}$, Karl M.H. Lai$^{2}$, Hegan Chen$^{1,2}$, Songqi Wang$^{1,2}$, Wei Xuan$^{1}$, \\ Ning Lin$^{2,4,\dag}$, Han Wang$^{2}$, Xiaojuan Qi$^{2,\dag}$ and Zhongrui Wang$^{1,3,\dag}$}

\affiliation{%
  \institution{\fontsize{8}{8}\selectfont$^1$ACCESS – AI Chip Center for Emerging Smart Systems, InnoHK Centers, Hong Kong Science Park, Hong Kong; $^2$Department of Electrical and Computer Engineering, The University of Hong Kong, Hong Kong; $^3$School of Microelectronics, Southern University of Science and Technology, Shenzhen, China; $^4$Department of Electrical Engineering, City University of Hong Kong, Hong Kong
  \\ $^*$Equal contribution to this work; $^\dag$Corresponding authors: linning@hku.hk; xjqi@eee.hku.hk; wangzr@sustech.edu.cn}
  \country{}
}

\renewcommand{\shortauthors}{K. Wong et al.}

\begin{abstract}
Edge Artificial Intelligence of Things (AIoT) systems often collect sensitive data \textit{in situ}, raising serious privacy concerns. Resistive-switching random-access memory (RRAM) is an attractive substrate for efficient AIoT thanks to its multi-bit storage and compute-in-memory (CiM) capabilities, while its inherently stochastic write behavior provides a natural source of randomness that can be leveraged for differential privacy (DP) protection. Yet how to transform this device-level randomness-- typically viewed as detrimental to accuracy-- into a principled randomized mechanism while preserving model utility remains underexplored.
We propose $\rramdp$, a hardware–algorithm co-design that relaxes RRAM write–verify operations to inject calibrated noise for inherently ($\varepsilon,\delta$)-DP with formal DP analysis; together with pretraining techniques, it renders a novel private, high-utility CiM-training paradigm. On CIFAR-10/100, STS-B, and SST-2, $\rramdp$-SGD incurs at best only a 3.8\% accuracy drop at ($\varepsilon$=2,$\delta$=O(1/n))-DP relative to non-private SGD. At the same privacy level, $\rramdp$-SGD delivers up to 57$\times$\& 3.2$\times$ energy savings and 2.7$\times$\& 1.8$\times$ speedups over A100 and DiVa-GEMM, respectively. These results point towards efficient, privacy-preserving in-memory training on RRAM at the edge.
\end{abstract}



\maketitle
\section{Introduction}
Edge Artificial Intelligence of Things (AIoT) systems increasingly operate \textit{in situ} sensing and learning from local, privacy-sensitive environments to personalize services across modalities and tasks~\cite{AIoT}. This shift heightens well-documented risks: seemingly anonymized records can be re-identified~\cite{linkage_attack01}, and modern deep neural networks (DNNs) memorize and leak training data via membership inference attacks (MIAs)~\cite{MIA01} and training data extraction attacks~\cite{data_extract_LLM,data_extract_Diffusion}. 

Differential privacy (DP) provides a principled defense by adding calibrated noise to query outputs, masking any single individual’s contribution and thereby offering rigorous protection against database-level attacks~\cite{privacybook}. DP-Stochastic Gradient Descent (SGD) extends DP guarantees to integrating noise injection in SGD optimizer, enabling privately trained DNNs with $(\varepsilon, \delta)$-DP~\cite{dpsgd}. However, despite their widespread use, existing DP and DP-SGD frameworks still face critical challenges in AIoT applications as detailed below.

To date, most DP mechanisms are realized in \emph{digital} hardware-- namely CPUs/GPUs/TPUs and dedicated digital accelerators~\cite{Opacus,diva}-- which raises several tensions intersected at efficiency and privacy for edge AIoT: (i) pseudo-random noise sampling in digital platform has finite precision with rounding errors, which is susceptible to floating-point error attacks \cite{floatingpt_attack01,floatingpt_attack02}; (ii) repeated data movement between processors and memory for real-time randomized mechanism are energy-expensive in local DP~\cite{LDP01} (iii) existing DP-SGD implementation libraries suffer $2\times -1000\times$ in time and space complexity compared to normal SGD training~\cite{bookkeeping_clip}, which are especially problematic in resource constrained edge.


\begin{figure}
    \centering
    {\includegraphics[width=1.0\columnwidth]{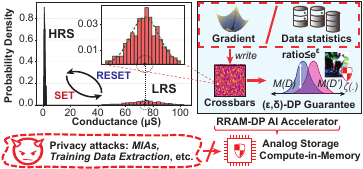}} \\
    \vspace{-8pt} 
    \caption{(LEFT) Histogram of cycle-to-cycle variation for an RRAM device between high resistance state (HRS) and low resistance state (LRS) with 300 iterative SET-RESET at fixed operation voltages. (RIGHT) Such noise can provide $(\varepsilon,\delta)$-DP for storing sensitive data statistics/model parameters.}%
    \label{fig:device}%
       \vspace{-10pt}  
\end{figure}

Resistive-switching random-access memory (RRAM) is promising solution to address these challenges: First, RRAM's inherently stochastic programming offers a native source of entropy for true random number generation. Second, write operations of RRAM, originally regarded as a reliability problem, are naturally an in-place randomization with \emph{\textbf{analog storage}} capability~\cite{RRAM_01}. Thirdly, the ohmic nature of RRAM naturally allows efficient multiply-accumulates (MACs) computation, where RRAM-based Compute-in-memory (CiM) accelerator provides powerful \emph{\textbf{in-memory vector-matrix multiplication}}. Prior DP work explored memristor noise largely at a theoretical level, under the assumption of idealized i.i.d. Gaussian perturbations~\cite{Memristor_DP}. In practice, RRAM noise is non-identical, and practically difficult to be controlled across devices and cycles for different privacy levels--whether how device-realistic randomness can be \emph{calibrated} to deliver formal DP guarantees~\cite{privacybook} remains open. In addition, CiM-training usually underperforms digital accelerators due to noise accumulating in write operations~\cite{ielmini2025resistive}. A scalable CiM-training paradigm is another challenge. 

To this end, we introduce $\rramdp$, which is a co-design that quantifies device randomness into privacy levels, along with pretraining techniques to boost CiM-training utility. To our knowledge, we propose the first RRAM-based randomized mechanism with $(\varepsilon, \delta)$-DP for both on-device data storage and SGD executed directly in RRAM through formal $\mu$-Gaussian DP ($\mu$-GDP) analysis grounded in \emph{\textbf{device statistics}}. For hardware, we realize a robust noise source at program time by relaxing write-verify cycles with only a single control variable $\tau$. The overview of $\rramdp$ is shown in Fig.~\ref{fig:device} and we summarize our contributions as follows:

\begin{itemize}

\item \textbf{First RRAM-based randomized mechanism.} Contrary to digital hardware with finite randomness and efficiency issues, we calibrate the privacy level of RRAM write operations for \textit{in situ} true noise addition $(\varepsilon, \delta)$-DP mechanisms, namely $\rramdp$ mechanism and $\rramdp$-SGD.

\item \textbf{A robust DP noise sampling from RRAM.} Challenges such as non-linear filament growth and fabrication non-idealities across devices are addressed by write-verify sampling and Lindeberg-Central Limit Theorem (CLT) design, calibrating steady noise with only one parameter.

\item \textbf{Noise-resilient CiM-training paradigm.} To mitigate the reliability degradation of noisy write operation, we connect pretraining techniques inspired by DP-SGD to \emph{in situ} RRAM-based CiM-training. Results show $\rramdp$-SGD incurs $\leq\sim$7\% accuracy drop at $\varepsilon$=5 compared to ideal training.


\item \textbf{State-of-the-art efficiency of DP-SGD accelerator.} Energy expensive DP operations can be locally processed and \emph{\textbf{stored}} by $\rramdp$-SGD accelerator to prevent privacy leakage. Simulations show a 54$\times$-57$\times$ and 3$\times$-3.2$\times$ energy savings; 2.5$\times$-2.7$\times$ and 1.7$\times$-1.8$\times$ speedup compared to A100 and digital DP-SGD accelerator DiVa-GEMM~\cite{diva}.

\end{itemize}

\section{Preliminaries}
\textbf{Problem Statement.} When AIoT parameters are not duly handled by privacy-preserving algorithms, an adversary can infer useful information through various privacy attacks. The least revealing attack is MIAs, which formulate a binary hypothesis test to decide if a datapoint exists in the training database~\cite{MIA01, MIA_database}. Model inversion attacks infer features that characterize each output class~\cite{Model_inversion}, making it possible to reconstruct training data, for example, if all class members depict the same person in a facial recognition model. The worst revealing attack is data extraction attack; generative models such as large language models (LLMs)~\cite{data_extract_LLM} and diffusion~\cite{data_extract_Diffusion} can memorize training data from which adversarially crafted prompts can expose individual’s bank account number, images, etc. These privacy attacks call for privacy-preserving techniques, namely randomized DP mechanisms, to protect privacy-sensitive edge parameters.

\subsection{\textbf{Standard RRAM cell}}
RRAM uses analog conductance $G$ to denote a real number $\theta$ via a mapping function $g:[\theta_{\min},\theta_{\max}]\rightarrow$$[G_{\min},G_{\max}]$. An example mapping function could be: 
\begin{equation}
  g(\theta) = G_{\min} + (\theta-\theta_{\min})\cdot\frac{G_{range}}{\theta_{range}}
 \label{eq:projection}
\end{equation}
where $G_{range}=G_{\max}-G_{\min}$, $\theta_{range}=\theta_{\max}-\theta_{\min}$.

However, writing a target conductance to RRAM suffers from inevitable randomness due to the underlying electrochemical ion migration \cite{rram_noise}. More specifically, there exist cycle-to-cycle variation and device-to-device variation that constitute a complex and dynamic RRAM noise model on a crossbar discussed in sec.3.

\textbf{Motivations.} We address the challenges of non-ideal noisy RRAM write by experimentally measured variability, and propose a principled method for sampling and calibrating device-realistic noise to satisfy the formal definition of DP in this work. By linking device-level statistics to privacy parameters, our approach converts unavoidable hardware non-idealities into controlled privacy guarantees without architectural modification.

\subsection{\textbf{Differential Privacy}}
DP is a robust privacy guarantee (usually $\varepsilon\leq$10~\cite{choose_epsi}) to protect all records in a database by estimating the worst-case output change from the inclusion or exclusion of one record in the database. The technical formulation of DP is as follows.
\begin{definition}[$(\varepsilon, \delta)$-Differential Privacy\cite{privacybook}] \label{def:DP}
Given a randomized mechanism $\mathcal{M}:\mathcal{D}\to \mathcal{O}$ is $(\varepsilon, \delta)$-DP if for any neighbouring dataset pairs $D$, $D'\in \mathcal{D}$ from the input domain, and for any measurable output subset $S \subseteq \mathcal{O}$:
\[
\Pr[\mathcal{M}(D) \in S] \leq e^\varepsilon \Pr[\mathcal{M}(D') \in S] + \delta
\]
\end{definition}
If $\delta = 0$, $\mathcal{M}$ satisfies $\varepsilon$-DP. The parameter $\varepsilon$ measures the maximum privacy loss of $\mathcal{M}$. While $\delta$ allows for a small probability of $\varepsilon$ failure, $\delta$ should be smaller than $1/|D|$ to prevent privacy breach for tail examples~\cite{dpsgd,privacybook}. This DP guarantee cannot be weakened by analyzing the DP mechanism output unless there exists additional access to the original data.

\begin{definition}[$\Delta(\zeta)$ Sensitivity\cite{sensitivity}] \label{def:sens}
Let $\zeta:\mathcal{D}\rightarrow \mathcal{O}^{d}$ be any $d$-dimensional query function. The $L_2$ sensitivity of function $\zeta$ for any neighbouring $D, D'\in \mathcal{D}$ can be defined:
\[
\Delta(\zeta)\triangleq \max_{D, D'} \|\zeta(D) - \zeta(D')\|_2 
\]
\end{definition}

\subsection{\textbf{Gaussian Differential Privacy}}
GDP has the tightest known bounds on composition for mechanisms that closely align with Gaussian noise, such as DP-SGD \cite{gaussDP_Deep}. Since hardware RRAM noise modeling is inherently noisy and dynamic compared to software DP deployment, we choose $\mu$-GDP as an analytical tool for tightest privacy estimation of calibrated RRAM noise. By definition, $\mu$-GDP implies a lossless conversion towards a collection of $(\varepsilon, \delta)$-DP guarantees.
\begin{definition}[Trade-off Function $T$~\cite{gaussDP}]
Let $\phi$ be a measurable rejection rule for distinguishing distributions $\mathcal{P}$ and $\mathcal{Q}$ on the same space; $\alpha_\phi$ and $\beta_\phi$ be the probabilities of an adversary making Type I and Type II errors, respectively. The trade-off function $T_{\mathcal{P}||\mathcal{Q}}(\alpha):[0,1]\rightarrow [0,1]$ is defined as:
\[
T_{\mathcal{P}||\mathcal{Q}}(\alpha)\triangleq inf\{\beta_{\phi}:\alpha_{\phi}\leq\alpha\}
\]
\end{definition}
The infimum is taken over all measurable rejection rules $\phi$. The function $T_{\mathcal{P}||\mathcal{Q}}$ characterizes the optimal trade-off between Type I and Type II errors for distinguishing $\mathcal{P}$ from $\mathcal{Q}$.

\begin{definition}[$\mu$-Gaussian Differential Privacy~\cite{gaussDP}]\label{def:Neyman_Pearson}
Given a randomized mechanism $\mathcal{M}$ is $\mu$-GDP (also known as $f$-DP, where $T(\mathcal{M}(D),\mathcal{M}(D')) \geq f$, $f$ is gaussian-parameterized) if for any neighbouring dataset pairs $D$, $D'\in \mathcal{D}$:
\[
T(\mathcal{M}(D),\mathcal{M}(D')) \geq T(\mathcal{N}(0,1),\mathcal{N}(\mu, 1))=G_\mu
\]

The exact solution of $G_\mu(\alpha)=\Phi(\Phi^{-1}(1-\alpha)-\mu)$ is derived by the Neyman–Pearson lemma \cite{Neyman_Pearson}, which states the most powerful test of differentiating two simple hypotheses is a likelihood ratio test. $\Phi$ denotes the cumulative density function of the Gaussian distribution, which applies to sec. 3.3. 
\end{definition}

\begin{theorem}[$(\varepsilon, \delta)$-DP Conversion\cite{gaussDP}]
For $\varepsilon \geq 0$, a randomized mechanism $\mathcal{M}$ is $\mu$-GDP if and only if $\mathcal{M}$ is $(\varepsilon, \delta)$-DP that follows: 
\[
\delta(\varepsilon) \triangleq \Phi(-\frac{\varepsilon}{\mu}+\frac{\mu}{2})-e^{\varepsilon}\Phi(-\frac{\varepsilon}{\mu}-\frac{\mu}{2})
\]\label{thm:conversion}
\end{theorem}
\vspace{-10pt} 
We denote satisfying $\mu$-GDP is the same as satisfying ($\varepsilon,\delta$)-DP in the following paragraphs for simplicity.

\section{RRAM-DP Formulation}
\textbf{Protection Overview.} 
$\rramdp$ is a co-design framework formulated under an RRAM-based randomized mechanism at one output parameter level to satisfy $\mu$-GDP. Firstly, a $k$-device configuration of RRAM crossbars is designed in which every analog output parameter is programmed on $k$ devices with shared write control and read-averaging circuitry. Secondly, the cycle-to-cycle variation is calibrated by the convergence of a write-verify algorithm at the target conductance within an adjustable absolute error bound $\tau$. For statistics storage, write operations inject noise with $\rramdp$ mechanism to provide DP guarantee. For DNNs, write operations in $\rramdp$-SGD provides DP guarantee for trained models. 


\subsection{Additive Noise Strategy} \label{Sec 3A}
To calibrate noise, this work follows a classic write-verify method that writes RRAM iteratively with a feedback loop control \cite{rram_program}. As described in Fig. \ref{fig:additive_noise_0}, this method reads an RRAM conductance in the current cycle and proceeds to the next SET/RESET cycle until hitting within stopping threshold. This method can steadily capture variation due to a bounded sampling region, 
disregarding overshooting/undershooting problems from the non-linear filamentary growth~\cite{ielmini2025resistive}. The SET/RESET program voltage adjustment rules vary based on device characteristics heuristically. 

\textbf{Cycle-to-cycle Variation.}
Assume an absolute conductance error $\Delta G=G_w-G_t$ is collected from differencing written conductance $G_w$ and target conductance $G_t$. We first denote that there are $N$ devices for storing $N$ parameters. Then, the final cycle-to-cycle variation of $\Delta G_i$ for each device $i$ sampled at the converged write-verify algorithm follows an observed distribution Dist$(\tau,G_t)$, which is dependent on the algorithmic stopping threshold $\tau$ and target conductance $G_t=g(\theta_i)$:
\vspace{-3pt} 
\begin{equation}
    \Delta G_i = \Delta \theta_i \cdot \frac{ G_{range}}{\theta_{range}}, \text{ }i=1,...,N
 \label{eq:error}
 \vspace{-5pt} 
\end{equation}
\[
   given \text{ } \Delta G_i\sim \text{Dist}(\tau,G_t),\text{ } \forall \tau \in \mathbb{R}^+,\text{ }G_{t}\in[G_{\min},G_{\max}]
\]

Eq.~\ref{eq:error} maps the program error of an RRAM device that contributes to numerical error by $G_{range}:\theta_{range}$ projection rule stated in Eq.~\ref{eq:projection}. As the program voltage aims exactly towards $G_t$, this usually gives a favourable DP noise source for $\mathbb{E}[\Delta G_i]\approx 0$. This work will prove the relaxation of the $G_t$ condition by our experiments afterwards for a simple yet effective noise calibration. Then, we can express the precise variance control only by $\tau$ for $(\varepsilon,\delta)$-DP requirement. Denoting the parameter error $\Delta\theta$ as a random variable in Eq. \ref{eq:var_error},
\begin{equation}
  \sigma_i(\Delta\theta) =\sigma_i(\Delta G)\cdot\frac{\theta_{range}}{G_{range}}, \text{ }i=1,...,N
 \label{eq:var_error}
\end{equation}

The standard deviation function $\sigma_i(.)$ on $\Delta\theta$ implies the upper bound of RRAM-controlled parameter error can go unbounded as the hyperparameter $\theta_{range}$ increases or $G_{range}$ decreases. As Fig. \ref{fig:additive_noise_0} shows, more cycles are required for the convergence towards a smaller stopping threshold with smaller variance. This suggests the challenge of device-calibrated $\rramdp$ lies in noise minimization rather than maximization. 
\begin{figure}
    \centering
    {\includegraphics[width=1.0\columnwidth]{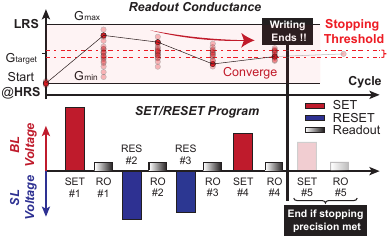}} \\
       \vspace{-10pt}  
    \caption{Control flow diagram for a classic RRAM write-verify algorithm. Voltage pulses into bit lines (BLs)/source lines (SLs) control SET/RESET, respectively.}
    \label{fig:additive_noise_0}%
       \vspace{-15pt}  
\end{figure}

\textbf{Device-to-device Variation.}
The non-ideal fabrication process, such as etching-induced damage and thickness of switching film~\cite{ielmini2025resistive}, causes each RRAM device to have a non-identical stochastic distribution of $\Delta G_i$$\sim$$\text{Dist}(\tau,G_t)$. Ideally, the cycle-to-cycle variation for every individual device should be independently modelled to ensure necessary noise is added without violating a strict $(\varepsilon,\delta)$-DP guarantee. Since auditing all devices is non-trivial, notably in very large crossbars, this work develops $\rramdp$ mechanism for device-to-device noise modeling relaxation below.

\subsection{Non-ideal Noise with Multi-device Representations}
\subsubsection{Smoothen Device-to-device Spread}
We denote $k\times N$ devices for $N$ parameters analog storage as $\Delta \widetilde{G}_i$$\sim$$\text{Dist}(\tau,G_t)$ to address device-to-device variation. Increasing $k$ reduces spread and smoothens outlier RRAM devices for a steady, consistent noise source at the cost of hardware.

\subsubsection{Unify Device-to-device Gaussian Shape} 
A $k$-device configuration design with averaged sum output satisfies the Lindeberg-CLT as $k$$\rightarrow $$\infty$~\cite{LindCLT}, as summing any independent random variable with bounded distribution and variance eventually converges to Gaussian. We formulate $\rramdp$ randomized mechanism for the statistical behaviour of a global non-uniform Gaussian write below:

\begin{theorem}
For some dataset $D\in \mathcal{D}$, we define an arbitrary query function $\zeta:\mathcal{D}\rightarrow\mathcal{O}^d$ and a sufficiently large k-device representation to write one parameter $\theta_i\in \zeta(D)$. Then, there exists an in situ $\rramdp$ mechanism $\mathcal{R}$ on the averaged sum output that satisfies $\mu_r$-GDP where: 
\label{mythm:rram_dp}
\begin{equation}
\mathcal{R}(D)\triangleq \zeta (D) + \mathcal{N}(0,\operatorname{diag}(\sigma_i(\Delta \theta)^2)) 
\label{eq:rram_dp}
\end{equation}
\[
\text{such that } \sigma_i(\Delta \widetilde{G})  \geq \frac{\Delta(\zeta)}{\mu_r} \cdot \frac{G_{\mathrm{range}}}{\theta_{\mathrm{range}}},\text{ } \forall i=\{1,2,...,N\}
\] 
\end{theorem}

\begin{proof}
Assume writing $\theta_i$ is strictly independent Gaussian due to CLT and physics randomness, querying $\zeta$ for two neighboring datasets $D\sim D'$ will be two multivariate Gaussian with $P \sim \mathcal{N}(m, \Sigma)$, $Q \sim \mathcal{N}(m', \Sigma)$ for writing all $\theta_i$. By Neyman-Pearson lemma of $\mu$-GDP in definition~\ref{def:Neyman_Pearson}, a likelihood ratio test $\Lambda(x)$ is considered for the privacy loss random variable:
\[
\Lambda(x) = \frac{f_P(x)}{f_Q(x)}=\frac{\text{exp}[-\frac{1}{2}(x-m)^{T}\Sigma^{-1}(x-m)]}{\text{exp}[-\frac{1}{2}(x-m')^{T}\Sigma^{-1}(x-m')]}
\]
\[
=\text{exp}[-\frac{1}{2}(x-m)^{T}\Sigma^{-1}(x-m)+\frac{1}{2}(x-m')^{T}\Sigma^{-1}(x-m')]
\]
\[
=\text{exp}[\Delta^{T}\Sigma^{-1}\cdot(x-\frac{m+m'}{2})] = \text{exp}[\Delta^{T}\Sigma^{-1}\cdot(\frac{\Delta}{2}+z)]
\]
where $\Delta =m-m'$; $x\sim P$, $z\sim\mathcal{N}(0,\Sigma)$.

Denote $\Delta^T\Sigma^{-1}z\sim\mathcal{N}(0,\Delta^T\Sigma^{-1}\Delta)$. The privacy loss of $\rramdp$ mechanism $\mathcal{R}$ can be shown as $\log \Lambda(x)|_ P\sim \mathcal{N}(\frac{1}{2}\Delta^T\Sigma^{-1}\Delta, \Delta^T\Sigma^{-1}\Delta)$. Similarly, $\mu$-GDP has a privacy loss of $\log \Lambda(x)|_{x\sim\mathcal{N}(\mu,1)}\sim\mathcal{N}(\frac{1}{2}\mu^2,\mu^2)$. Therefore, the goal is to identify the maximum $\mu=\sqrt{\Delta^T\Sigma^{-1}\Delta} \text{ } \forall D\sim D'$ such that $||\Delta||_2 \leq\Delta(\zeta)$ to declare $\mu$-GDP for the process. This yields a Rayleigh-quotient maximization problem:
\[
\max_{||\Delta||_2 \leq\Delta(\zeta)} \Delta^T\Sigma^{-1}\Delta= \Delta(\zeta)^2\lambda_{\max}(\Sigma^{-1})
\]
where $\lambda_{\max}(\Sigma^{-1})$ is the largest eigenvalue of $\Sigma^{-1}$.

This gives $\Delta(\zeta)\sqrt{\lambda_{\max}(\Sigma^{-1})}\leq\mu_{r}$ to claim $\mu_r$-GDP. With $\Sigma=\operatorname{diag}(\sigma_i(\Delta \theta)^2)$, $\lambda_{\max}(\Sigma^{-1})=\max _{i}{(1/\sigma_i^2)}={1/\sigma_{\min}^2}$. Therefore, we have $\sigma_{\min}(\Delta\theta)\geq\frac{\Delta(\zeta)}{\mu_r} \Rightarrow \sigma_{\min}(\Delta \tilde{G})\geq
\frac{\Delta(\zeta)}{\mu_r}\cdot\frac{G_{range}}{\theta_{range}}$ by substituting $\lambda_{\max}(\Sigma^{-1})$, which implies the conductance inequality $\forall i$.
\end{proof}

\subsection{RRAM-DP-Stochastic Gradient Descent}
Most RRAM CiM-accelerators only function as inference instead of training due to accumulative write variance during weight update~\cite{ielmini2025resistive, snngx}. Such RRAM write noise has a natural fit with the iterative noise addition in DP-SGD, which protects training data from being released in the privately trained model with $(\varepsilon,\delta)$-DP at the edge. We formulate $\rramdp$-SGD over the composition of $T$ iterative $\rramdp$ mechanism on the gradient function with private subsampling amplification in Algo.~\ref{alg:rramsgd}:
\vspace{-6pt}
\begin{algorithm}
\caption{$\rramdp$-SGD}
\setlength{\floatsep}{4pt}
\setlength{\textfloatsep}{4pt}
\begin{algorithmic}[1]
\State \textbf{Input:} Private examples $D = \{x_j \mid j = 1, 2, \ldots, n\}$,
\Statex Loss function $\mathcal{L}(\theta , x)$. Parameters: learning rate $\eta$, noise model $\sigma'_{\min}$, clip norm $C$, batch size $B$, iterations $T$
\For {$t=0,1,\ldots,T-1$}
    \State \textbf{Random-subsample $D_b$ uniformly}
    \Statex \hspace{1.2em} where $D_b \subset D, \quad D_b = \{x_j \mid j \in I_t\}, \quad |I_t| = B$
    
    \For {each $j \in I_t$}
        \State \textbf{Compute gradient: (RRAM)} 
        \Statex \hspace{2.7em} $g_{x_j} \gets \nabla_\theta \mathcal{L}(\theta_t , x_j)$
        \State \textbf{Clip gradient:} 
        \Statex \hspace{2.7em} $\bar{g}_{x_j} \gets g_{x_j} \cdot \min(1, C / \|g_{x_j}\|_2)$
    \EndFor
    \State \textbf{Average gradient and descend:}
    \Statex \hspace{1.2em} $\bar{\theta}_{t+1} \gets \theta_t - \eta \cdot \frac{1}{B} \sum_{j\in I_t} \bar{g}_{x_j}$
    \State \textbf{Projection ($| \mathbf{\bar{\theta}_{t+1}}|=N$) \& Write: (RRAM)} 
    \Statex \hspace{1.2em} $\sigma_i(\Delta \widetilde{G})\geq  \sigma'_{\min}\cdot\frac{2C\eta}{B}\cdot\frac{G_{range}}{\theta_{range}},\text{ }\forall i=\{1,2,,...,N\}$ 
    \State \textbf{Estimate privacy} $\mu_r$-GDP at iteration $t+1$ \textbf{(Eq. 5)}
\EndFor
\State \textbf{Output:} $\theta_T$ with privacy budget $(\varepsilon,\delta)$-DP conversion 
\end{algorithmic} 
\label{alg:rramsgd}
\end{algorithm}
\vspace{-12pt}
\begin{theorem}
Let $\mathcal{C}_p(\mathcal{R})^{\otimes T}$ denote the $T$-fold composition of \\ $\rramdp$ mechanism $\mathcal{R}$ under uniform subsampling rate $p = \frac{B}{n}$. Define $\sigma_{\min}' \leq \inf_{1 \le t \le T} \sigma_{\min,t}'$ for all iteration $t$. Then, in the asymptotic regime $p\sqrt{T}\rightarrow$ $c$ as $T\rightarrow \infty$, the $\rramdp$-SGD algorithm of $\mathcal{C}_p(\mathcal{R})^{\otimes T}$ with per-step noise multiplier $\operatorname{diag}(\sigma'_{i,t})$ \label{thm:DPSGD} is $\mu_r$-GDP:
\begin{equation}
\mu_{r}= p\cdot\sqrt{2T\cdot (e^{\sigma_{\min}^{'-2}}\cdot\Phi(1.5\sigma_{\min}^{'-1})+3\Phi(-0.5\sigma_{\min}^{'-1})-2)}
\end{equation}
\[
\text{such that } \sigma_i(\Delta \widetilde{G})\geq  \sigma'_{\min}\cdot\frac{2C\eta}{B}\cdot\frac{G_{range}}{\theta_{range}},\text{ }\forall i=\{1,2,\dots,N\}
\]
\end{theorem}
\vspace{-12pt}
\begin{proof}
In iteration $t$, given lines 6 and 8 for any two neighbouring datasets
$D \sim D'$, the per-iteration gradient update $\zeta(D) = \frac{\eta}{B}\sum_{j\in I_t} \bar{g}_{x_j}$ 
has sensitivity bound $\Delta(\zeta) = \frac{2C\eta}{B}$. In line 9, $\rramdp$-SGD adds calibrated noise $\Delta \theta \sim \mathcal{N}\!\left(0, \Sigma_t\right), \Sigma_t = \operatorname{diag}\bigl((\Delta(\zeta)\,\sigma_{i,t}')^2\bigr)$, so that the $i$-th coordinate noise standard deviation is
$\Delta(\zeta)\,\sigma_{i,t}'$. Let $\sigma_{\min,t}' = \min_i \sigma_{i,t}'$. By theorem~3.1 ($\rramdp$ 
mechanism GDP guarantee), the mechanism at iteration $t$ is $\mu_t$-GDP with $\mu_t = \frac{1}{\sigma_{\min,t}'}$. Since sampling all device true distributions is non-trivial, we model an empirical statistical lower bound function with $\sigma'_{\min}$ (see sec. 4.3). Then for all $t$, we have $\sigma_{\min,t}' \ge \sigma_{\min}' 
\Rightarrow
\mu_t = \frac{1}{\sigma_{\min,t}'} \;\le\; \frac{1}{\sigma_{\min}'}$. We define the actual worst-case per-step GDP parameter $\mu_{\max}$:
\[
\mu_{\max} \triangleq \sup_{1 \le t \le T} \mu_t \;\le\; \frac{1}{\sigma_{\min}'}
\]
Let $G_{\mu}$ denote the trade-off function of a $\mu$-GDP Gaussian
mechanism. The composition rule for GDP ~\cite{gaussDP} states that composing
$\mu_t$-GDP mechanisms gives $ G_{\mu_1} \otimes \cdots \otimes G_{\mu_T} = G_{\sqrt{\mu_1^2 + \cdots + \mu_T^2}}$. 

Since $\mu_t \le \mu_{\max}$ for all $t$,
\[
\sqrt{\sum_{t=1}^T \mu_t^2} 
\;\le\; \sqrt{T \mu_{\max}^2} 
\;=\; \sqrt{T}\,\mu_{\max}
\;\le\; \frac{\sqrt{T}}{\sigma_{\min}'}
\]

Therefore, the $T$-fold (non-subsampled) composition satisfies
\[
G_{\mu_1} \otimes \cdots \otimes G_{\mu_T}
= G_{\sqrt{\sum_{t=1}^T \mu_t^2}}
\succeq G_{\sqrt{T}\,\mu_{\max}}
\succeq G_{\sqrt{T}/\sigma_{\min}'}
\]

Equivalently, we can model each iteration as a
$\mu'$-GDP mechanism with $\mu' = \frac{1}{\sigma_{\min}'}$ as the base mechanism with trade-off function $G_{\mu'}$; thus, the $T$-fold composition 
$\mathcal{C}_p(\mathcal{R})^{\otimes T}$ with privacy amplification convergence through batch subsampling rate $p$ in the asymptotic regime is no worse than $\mu_r = p \sqrt{2T \,\chi_+^2(G_{\mu'})}$ as $p\sqrt{T} \to$ constant $c$, where $\chi_+^2(G_{\mu}) 
= e^{\mu^2}\Phi(3\mu/2) + 3\Phi(-\mu/2) - 2$ under GDP Lemma~5.3~\cite{gaussDP}.
Therefore, substituting $\mu' = \frac{1}{\sigma_{\min}'}$ yields $\mu_{r}= p\cdot\sqrt{2T\cdot (e^{\sigma_{\min}^{'-2}}\cdot\Phi(1.5\sigma_{\min}^{'-1})+3\Phi(-0.5\sigma_{\min}^{'-1})-2)}$.
\end{proof}

In Algo.~\ref{alg:rramsgd}, we accelerate $\rramdp$-SGD in an analog-digital computing fashion, where gradient computation (line 5) and noise generation (line 9) can be directly processed by analog RRAM crossbar. Since RRAM noise occurs at parameter level instead of gradient, we reverse the normal gradient descent compute logic of DP-SGD in lines 8 and 9, which causes the privatised gradient sensitivity $\Delta(\zeta)$ to scale with learning rate $\eta$ and batch size $B$. Lastly, we propose an $\rramdp$-SGD accountant (Eq. 5) in theorem~\ref{thm:DPSGD} to rigorously track device non-idealities with $\sigma'_{\min}$ after training iteration $t$ for line 9. RRAM-DP-SGD is defined under the replacement definition with 2 times the sensitivity of the add-or-remove definition~\cite{dp_addremove}.


\subsection{Pretraining the CiM-based Training}
The program precision from memristor non-idealities causing significant degradation in network performance (e.g., $\geq$60\% accuracy drop on CIFAR10/100~\cite{eapu}) is a major bottleneck for memristor-based training~\cite{ai_accelerator02,pipeline}. While this work rethinks the formal connection between device-realistic noise and theoretical DP-SGD (noisy training), we find an elegant way of borrowing DP-SGD training techniques to improve utility in CiM-training. Due to the noisy privatised mechanism, DP-SGD has limited utility as compared to normal SDG training. One of the scalable methods for improving DP-SGD utility is through pretraining non-privately on massive public data or synthetic data, and fine-tuning with DP-SGD on precious private data~\cite{DPSGD21,DPSGD22,DPSGD23}. We believe more scaling techniques from DP-SGD being migrated to improve analog CiM-training are worth future research exploration despite the differences in calibrating device-realistic noise; and we adopt a pretraining strategy from DP-SGD in this work~\cite{DPSGD22}, such as a larger batch size and group normalization layer, to improve general CiM-training utility.

\section{Experiments}
\subsection{Experimental Setup}
\textbf{Memory Devices.}
The electrical characterization is performed using a test chip fabricated in a 180~nm process, which integrates a 32$\times$32 one-transistor-one-resistor (1T1R) macro. The memory cells comprise TiN as bottom and top electrodes with a Ta$_2$O$_5$/TaO$_x$ dielectric stack. The 1T1R array employs a crossbar topology, with row-shared word lines (WLs) and source lines (SLs) and column-shared bit lines (BLs) for efficient cell addressing.

\textbf{Hybrid Analog-digital System.} We perform hybrid analog-digital training from scratch on MNIST (8$\times$8 resized) with a 2-layer MLP network of 760 weights. Weights are stored in analog RRAM arrays, whereas dynamic parameter updates are implemented in a digital computing unit with
SRAM buffer. The minimum viable testing uses hyper-parameters $\delta$=$10^{-5}$, batch size=2048, clip norm=2, LR=2, [$G_{\min}$, $G_{\max}$]=[20, 80], [$\theta_{\min}$, $\theta_{\max}$]=[-1, 1].

\textbf{Simulation Benchmarks.}
We scale up $\rramdp$-SGD on Opacus~\cite{Opacus} with experimentally acquired RRAM noise statistics on the 32$\times$32 1T1R macro for CIFAR-10, CIFAR-100~\cite{C10_dataset}, STS-B, and SST-2~\cite{glue} datasets. Following~\cite{DPSGD22}, we use pretrained models of WRN-16-4, WRN-28-10~\cite{WRN}, BERT-Base~\cite{BERT}, and RoBERTa-Base~\cite{Roberta}, respectively. Table \ref{tab:training} shows detailed private training hyper-parameters simulated by  NVIDIA RTX 4090 GPU.

\begin{table}[!h]
\vspace{-5pt}
\centering
\caption{Simulation settings for $\rramdp$-SGD}
 \vspace{-10pt}
\resizebox{1\linewidth}{!}{
\Large 
\begin{tabular}{p{3cm}c|c|c|c}
\toprule
\multirow{2}{*}{} & \multicolumn{4}{c}{\textbf{Datasets}} \\ \cline{2-5}
& \cellcolor[HTML]{EFEFEF} \textbf{CIFAR10}  & \cellcolor[HTML]{EFEFEF} \textbf{CIFAR100}  & \cellcolor[HTML]{EFEFEF} \textbf{STS-B} & \cellcolor[HTML]{EFEFEF} \textbf{SST-2} \\ 
\toprule
\toprule
\cellcolor[HTML]{EFEFEF} \textbf{Length (n)} & \cellcolor[HTML]{EFEFEF} \textbf{50k} & \cellcolor[HTML]{EFEFEF} \textbf{50k} & \cellcolor[HTML]{EFEFEF} \textbf{7k} & \cellcolor[HTML]{EFEFEF} \textbf{67k} \\
\cellcolor[HTML]{EFEFEF} $\boldsymbol{\delta}$ & \cellcolor[HTML]{EFEFEF} $\mathbf{10^{-5}}$ & \cellcolor[HTML]{EFEFEF} $\mathbf{10^{-5}}$ & \cellcolor[HTML]{EFEFEF} $\mathbf{10^{-4}}$ & \cellcolor[HTML]{EFEFEF} $\mathbf{10^{-5}}$ \\
Pretraining  & ImageNet32~\cite{Imagenet32} & ImageNet32~\cite{Imagenet32} & Huggingface~\cite{BERT} & Huggingface~\cite{Roberta} \\
Fine-tuning & Entire Model & Entire Model & LoRA(r=2) & LoRA(r=4)  \\
\midrule
\midrule
\cellcolor[HTML]{EFEFEF} \textbf{Batch Size (B)} & \cellcolor[HTML]{EFEFEF} \textbf{2048} & \cellcolor[HTML]{EFEFEF} \textbf{2048} & \cellcolor[HTML]{EFEFEF} \textbf{256} & \cellcolor[HTML]{EFEFEF} \textbf{2048} \\
Clip Norm (C) & 2 & 1 & 1.5 & 2  \\
LR ($\eta$) & 1.5 & 1.5 & 0.5 & 1.5  \\
$\text{[}G_{\min}, G_{\max}\text{]}$  & $\text{[20}, \text{80]}$ & $\text{[20}, \text{80]}$ & $\text{[20}, \text{80]}$ & $\text{[20}, \text{80]}$  \\
$\text{[}\theta_{\min}, \theta_{\max}\text{]}$  & $\text{[-1}, \text{1]}$ & $\text{[-2}, \text{2]}$ & $\text{[-1}, \text{1]}$ & $\text{[-1}, \text{1]}$  \\
\bottomrule
\end{tabular}}
\label{tab:training}
\vspace{-8pt}
\end{table}

\textbf{Hardware Accelerator.} Our hardware accelerator shown in Fig.~\ref{fig:accelerator} borrows a general RRAM-based CiM architecture~\cite{ai_accelerator02} to implement $\rramdp$-SGD. The system core consists of a global control unit and 512 tiles; each tile contains 16$\times$16 processing elements (PEs). Each PE contains three 32$\times$32 RRAM arrays, where a weight value is stored by a 3-device configuration. Along with peripheral circuits such as digital-to-analog converters (DACs), analog-to-digital converters (ADCs), and read/write drivers, it performs energy-efficient in-memory vector-matrix multiplication. In Fig.~\ref{fig:accelerator}, the 3-device configuration shares the same DACs, ADCs, and drivers; and their output currents are summed along the shared SL during readout.

\begin{figure}
    \centering
    {\includegraphics[width=0.95\columnwidth]{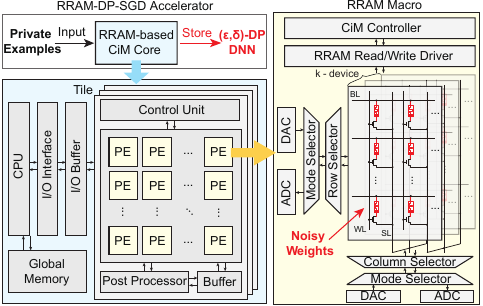}} \\
    \vspace{-8pt}
    \caption{Architecture of an RRAM-based CiM accelerator for $\rramdp$-SGD training, including the system-level framework and the circuit-level configuration of a single PE for $k$ RRAM macros with peripheral circuits.}%
    \label{fig:accelerator}%
    \vspace{-5pt}  
\end{figure}

\textbf{Evaluation Metrics.}
Model utility: We compare test accuracy of $\rramdp$-SGD with non-private SGD training. Minimal performance drop reports good utility. 
Model training cost: We evaluate FLOPS~\cite{calflops} for one iteration.
Hardware performance: We compare energy and time cost over FP32 training on A100 GPU (40GB) at peak 19.5TFLOPS, maximum power 250W; and DiVa-GEMM at peak 29.5TFLOPS, maximum power 21.2W~\cite{diva}. 

\begin{figure}
    \centering
    {\includegraphics[width=0.85\columnwidth]{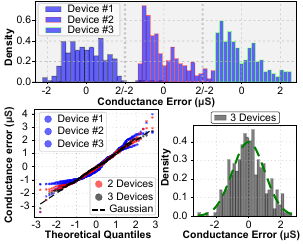}} \\
       \vspace{-10pt}  
    \caption{(TOP) Histogram of 300 samples from 3 independent RRAM devices, sampled at $\Delta G_{i=1,2,3}\sim \text{Dist}(\tau\text{=3}\text{\textmu S},G_t\text{=50}\text{\textmu S})$, respectively. (BOTTOM) Q-Q plot for the averaged sum \& Histogram of the 3-device representation.}
    \label{fig:3-device}%
    \vspace{-15pt}  
\end{figure}

\vspace{-3pt}
\subsection{Robust Gaussian for 3-Device Representation}
We first study the sufficient $k$-device representation to generate Gaussian noise. In Fig~\ref{fig:3-device} (TOP), the graph shows the histogram for 3 single independent RRAM devices sampled at $G_t\text{=50}\text{\textmu S} \text{ and } \tau\text{=3}$\textmu S for 300 times. These 3 distributions follow a similar spread centred around 0\textmu S but a randomly different shape, which presents the independent, non-identically distributed challenge of device-to-device variation for completing a theoretical DP guarantee proof. The bottom left plot shows the Q-Q plot as $k$ increases from 1 to 3 upon the averaged sum of $3Ck$ devices. It is shown that a 3-device representation aligns nearly perfectly with the diagonal Gaussian under our write-verify sampling method with empirical RRAM stochastic characteristics. The bottom right histogram also expresses the normal Gaussian curve comparison, which says the 3-device representation is a good trade-off between tight privacy evaluation and hardware costs of duplicating $k$ RRAM devices. We confirm this Gaussianity condition ($k$=3) by scaling to $\tau$=1\text{\textmu S}, 10\text{\textmu S} and $G_t$=20\text{\textmu S}, 80\text{\textmu S} in our empirical qualitative experiments.

\begin{figure}[!t]
    \centering
    {\includegraphics[width=0.95\columnwidth]{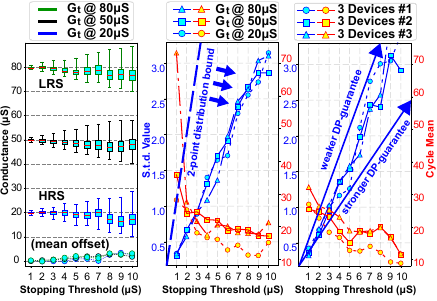}} \\
       \vspace{-10pt}  
    \caption{(LEFT: $\rramdp$ Robustness) Boxplot for the write-verify noise sampling with 3-device representations for a 32$\times$32 crossbar at 20$\text{\textmu S}$, 50$\text{\textmu S}$, 80$\text{\textmu S}$. (MID: $\rramdp$ Sampling) Standard deviation (s.t.d.) and program cycles at 20$\text{\textmu S}$, 50$\text{\textmu S}$, 80$\text{\textmu S}$. (RIGHT: $\rramdp$ Modeling) S.t.d. and program cycles for 3 random 3-device representations at 50$\text{\textmu S}$.}
    \label{fig:linearity}%
    \vspace{-25pt}  
\end{figure}

\vspace{-3pt}
\subsection{RRAM-DP Relaxation on \texorpdfstring{$G_t$}{Gt} \& Device Variation}
In Fig.~\ref{fig:linearity}, we show the stochastic pattern and the correct form of $\rramdp$ implementation for a 32$\times$32 1T1R RRAM crossbar. In Fig.~\ref{fig:linearity} (LEFT), a 3-device representation is sampled across the macro at 20$\text{\textmu S}$, 50$\text{\textmu S}$, 80$\text{\textmu S}$ upon the convergence of the write-verify algorithm. The box-plot shows consistent Gaussian-like properties at different resistance states, such as symmetry and stable interquartile range spread, over varying $\tau$. The mean offset is also consistent across 20$\text{\textmu S}$, 50$\text{\textmu S}$, 80$\text{\textmu S}$, which is negligible when $\tau$ is small. In Fig.~\ref{fig:linearity} (MID), the standard deviation of $\Delta \widetilde{G}_i$$\sim$$\text{Dist}(\tau,G_t)$ is found to have a favourable \emph{\textbf{linear}} relationship with varying $\tau$ bounded by a 2-point distribution. More importantly, this linearity is consistent and \emph{\textbf{independent of all resistance states}}. This implies $\Delta \widetilde{G}_i$$\sim$$\text{Dist}(\tau,G_t)$ can be relaxed to $\Delta \widetilde{G}_i$$\sim$$\text{Dist}(\tau)$ for a simple $\rramdp$ sampling that depends only on $\tau$. In Fig.~\ref{fig:linearity} (RIGHT), we demonstrate the hardware lower bound modeling $\sigma(\Delta \widetilde{G}_\mathrm{LB})$ for $\sigma'_{\min}$ in $\rramdp$ mechanism $\mathcal{R}$. We initialize fixed sampling windows for all devices of 300 empirical error entries at 50$\text{\textmu S}$. Empirical s.t.d. are then estimated over these windows using Chi-Square method ~\cite{LindCLT} ($\alpha$=0.05) for Gaussian s.t.d. lower bound. We show that individual-level sampling with three examples of 3-device representations follows \emph{\textbf{independent linear curves}} with slightly different slopes. The linear overlapping region suggests recalibration is only required if any observed window drops below the maintained lower bound. We also report around 10 to 70 expected cycles per device needed for our write-verify algorithm to converge in red color. Cycles may increase exponentially as $\tau$ approaches from 1 to 0. 

\begin{figure}
    \centering
    {\includegraphics[width=0.95\columnwidth]{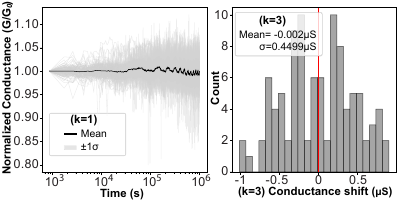}} \\
    \vspace{-10pt}  
    \caption{(LEFT) Device drift from 15 mins to 12 days. (RIGHT) Final equivalent drift error for 3-device representations.}
    \label{fig:drift}%
    \vspace{-10pt}  
\end{figure}

\begin{figure}
    \centering
    {\includegraphics[width=1.\columnwidth]{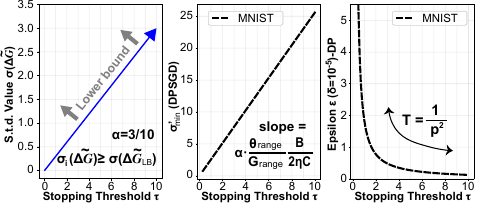}} \\
    \vspace{-8pt}  
    \caption{(LEFT) Hardware lower bound modeling function. (MID) Software equivalent lower bound function $\sigma'_{\min}(\tau)$ on MNIST dataset. (RIGHT) Final equivalent privacy trade-off function $\varepsilon(\tau)$ computed at $\delta$=$10^{-5}$ on MNIST dataset.}
    \label{fig:mnist_model}%
       \vspace{-15pt}  
\end{figure}

\vspace{-3pt}
\subsection{Impact to RRAM-DP from Device Drift}
Device drift in RRAM refers to the gradual change in resistance state over time after programming, typically attributed to the evolution of conductive filamentary microstructure~\cite{rram_noise}. In Fig.~\ref{fig:drift} (LEFT), we perform an experiment to study the effects of individual device drift after $\rramdp$ mechanism under room temperature. Results show that device conductance ($k$=1) normalised to initial states experienced minimal device drift after $10^6$~s, with mean$\approx$0 and spread$\approx\pm$0.05. In Fig.~\ref{fig:drift} (RIGHT), we report the final--initial drift error under 3-device representation ($k$=3) for $10^6$~s. This shows device drift results in an additional Gaussian noise injection with mean$\approx$0$\text{\textmu S}$ and standard deviation$\approx$0.45$\text{\textmu S}$, which is equivalent to $\rramdp$ modeling at $\tau$=1.5$\text{\textmu S}$ from Fig.~\ref{fig:linearity} (RIGHT). This suggests that device drift is introducing an additional one-shot noise offset to $\rramdp$ mechanism $\mathcal{R}$ at the quasi-saturation level over time. While drift error is small to be considered for one-shot storage mechanism $\mathcal{R}$, it becomes negligible to $\rramdp$-SGD as iteration $T$ is large enough for the accumulative program error $\gg$ drift error. From the privacy perspective, the equivalent per-step drift noise in $\rramdp$-SGD provides negligible additional privacy amplification effect over fixed gradient sensitivity as $T$ increases.  

\vspace{-3pt}
\subsection{Hybrid Analog-digital Training}
To validate device-calibrated DP, we perform hybrid $\rramdp$-SGD training on MNIST dataset, where digital update signals are translated into programming pulses to modulate RRAM devices, enabling in-situ ($k$=3) analog weight adaptation.

\textbf{RRAM-DP Modeling.} In Fig.~\ref{fig:mnist_model} (LEFT), we first model the linear hardware lower bound function $\sigma(\Delta \widetilde{G}_\mathrm{LB})$ initialized by Fig.~\ref{fig:linearity} (RIGHT) at slope $\alpha$=3/10. At runtime, the sampling window remove oldest error sample and append new final read-out error at each training iteration. In Fig.~\ref{fig:mnist_model} (MID), we transform the hardware lower bound into equivalent software lower bound function $\sigma'_{\min}$ that is scaled by MNIST dataset sensitivity, such as batch size and learning rate, for hardware $\tau$ as a control variable. In Fig.~\ref{fig:mnist_model} (RIGHT), we further compute the ($\varepsilon,\delta$)-DP form by transforming $\sigma'_{\min}$ into $\mu_r$-GDP using $T=1/p^2$ as the correct convergence scaling from $\rramdp$-SGD in theorem 3.2 and ($\varepsilon,\delta$)-DP conversion in theorem 2.1. Thus, by choosing any write-verify control variable $\tau$, we get the desirable $\rramdp$-SGD privacy level calibrated for MNIST. 

\vspace{2pt}
\textbf{RRAM-DP-SGD Training.} In Fig.~\ref{fig:mnist_training} (LEFT), we perform hybrid analog-digital training with $\rramdp$-SGD and compare the actual hardware noise with simulation noise, DP-SGD training, and normal training using a clipping-free and noise-free optimizer. While the normal training is capped at $\sim$90\% accuracy, we train privately using noise calibrated at $\varepsilon$=0.1. First, DP-SGD uses the standard $\mu$-GDP accountant~\cite{gaussDP} to estimate isotropic Gaussian noise injection. Second, $\rramdp$-SGD Hardware uses $\tau\approx$10 from Fig.~\ref{fig:mnist_model} (RIGHT) to inject realistic RRAM noise at $\varepsilon$=0.1. Third, $\rramdp$-SGD Simulation uses $\tau$=10 with empirical s.t.d. collected for all devices in Fig.~\ref{fig:linearity} (RIGHT) to inject diagonal Gaussian noise with a fixed s.t.d. lookup table. Results show that $\rramdp$-SGD simulation closely mimics the $\rramdp$-SGD hardware training behaviour, both achieving accuracy at $\sim$80\%. This also shows $k$=3 closely approximates the theoretical Gaussian noise. Meanwhile, DP-SGD achieves similar accuracy of $\sim$82\% at the same theoretical privacy level, which proves the tight hardware estimation with $\tau\approx$10.

\vspace{2pt}
\textbf{Training Program Cycles.} In Fig.~\ref{fig:mnist_training} (RIGHT), we record the total programming cycles of $\rramdp$-SGD along the hardware training iterations. Given 760 weights in the MLP network, roughly 21000--23000 cycles are performed in each iteration. Consider the 3-device representation, we randomly sample 3 devices for the same weight to program, where each device is expected to program 23000/(760*3)$\approx$10.1 times. We find this result closely aligns with program cycles shown in $\rramdp$ Sampling in Fig.~\ref{fig:linearity} when $\tau$=10. 

\begin{figure}[!t]
    \centering
    {\includegraphics[width=1\columnwidth]{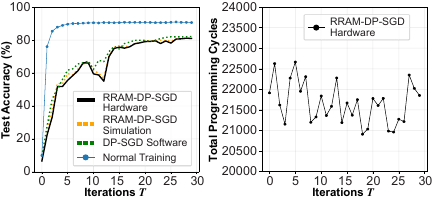}} \\
       \vspace{-10pt}  
    \caption{Experimental results for $\rramdp$-SGD hardware training ($\varepsilon$=0.1) as compared to DP-SGD, normal training.}
    \label{fig:mnist_training}%
    \vspace{-15pt}  
\end{figure}

\vspace{-2pt}
\subsection{Privacy Protection against MIAs}
In Fig.~\ref{fig:attack}, we identify the privacy protection power between normal training without privacy-preserving techniques, DP-SGD, and $\rramdp$-SGD hardware all trained in Fig.~\ref{fig:mnist_training}. For the MIA attack, we use Likelihood Ratio Attack (LiRA)~\cite{MIA01} with 256 shadow models (N=256) performed under 10,000 balanced attack samples to capture membership and non-membership training distribution on the resized 8$\times$8 MNIST dataset. The effectiveness of the LiRA attack can be evaluated from two perspectives: (1) an area under the ROC curve (AUC) close to 0.5, indicating near-random discrimination capability; and (2) a low true positive rate (TPR) at low false positive rates (FPR), reflecting limited statistical power when the significance level is small. As shown in Fig.~\ref{fig:attack}, normal training exhibits mild privacy leakage (AUC = 0.5348). In contrast, both DP-SGD and $\rramdp$-SGD achieve AUC values close to 0.5 and suppressed TPR in the low-FPR regime. This confirms the strong protection power of utilizing device-realistic noise for DP private learning.

\begin{figure}
    \centering
    {\includegraphics[width=1\columnwidth]{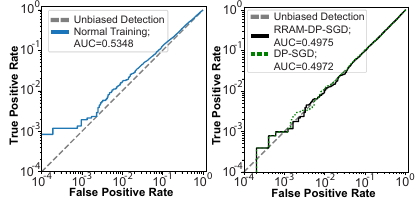}} \\
    \vspace{-10pt}  
    \caption{ROC curves of MIA (LiRA (N=256)) against normal training, DP-SGD, and $\rramdp$-SGD hardware training calibrated at ($\varepsilon$=0.1, $\delta$=$\mathbf{10^{-5}}$)-DP protection level.}
    \label{fig:attack}%
    \vspace{-7pt}
\end{figure}

\begin{figure}
    \centering
    {\includegraphics[width=1\columnwidth]{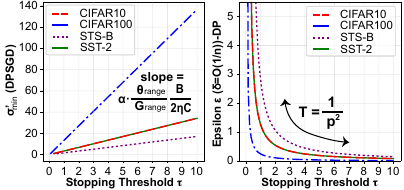}} \\
    \vspace{-7pt}  
    \caption{(LEFT) Software equivalent lower bound function $\sigma'_{\min}(\tau)$ on vision/language datasets. (RIGHT) Final equivalent privacy trade-off function $\varepsilon(\tau)$ computed at $\delta$=O(1/n) on vision/language datasets.}
    \label{fig:rram_dp}%
    \vspace{-10pt}  
\end{figure}

\vspace{-4pt}
\subsection{Unlocking CiM-based Training with Privacy and Pretraining Techniques}
\textbf{RRAM-DP Modeling.} To fully utilize noise of \textit{in situ} CiM-training, we scale up $\rramdp$-SGD training from MNIST to CIFAR10/100 (vision) and STS-B/SST-2 (language) tasks. Similarly, we identify the same linear lower bound model of $\sigma(\Delta \widetilde{G}_\mathrm{LB})$ = $\alpha\cdot\tau$ with slope $\alpha$ = 3/10 derived from Fig.~\ref{fig:linearity} to estimate $\sigma'_{\min}$ for computing the composition of $\mu_r$ by theorem~\ref{thm:DPSGD}. Here, a steeper slope of lower bound $\sigma'_{\min}$ refers to a stricter privacy guarantee at the same hardware noise injection as dataset-specific training hyper-parameters scale the gradient sensitivity. Then, we formulate $\varepsilon$ ($\delta$ follows Table~\ref{tab:training}) against $\tau$ relationship in Fig.~\ref{fig:rram_dp} (RIGHT) to calibrate target privacy ($\varepsilon,\delta$)-DP at some $\tau$ for different dataset training configurations. It is observed that the flatter slope in Fig.~\ref{fig:rram_dp} (LEFT) has easier privacy control in Fig.~\ref{fig:rram_dp} (RIGHT) with $\rramdp$-SGD over four datasets. 

\vspace{2pt}
\textbf{Pretraining Effects.} Usually, a larger dataset length or a dataset with more classes, such as CIFAR100, requires a richer visual pretraining dataset and larger batch size to improve utility, where the later one causes the steep lower bound $\sigma'_{\min}$ slope to be more difficultly compared at the same $\varepsilon$ across datasets. Here, we adopt the same convergence scaling of $T$=$1/p^2$ from theorem \ref{thm:DPSGD} by fixing the batch size to compute the iterations $T$ needed for training and replacing batch normalization layers with group normalization for per-sample clipping. Upon these setups, we simulate the CiM noise in the same way from Fig.\ref{fig:linearity} and Fig.\ref{fig:mnist_training}  for privately training CIFAR10, CIFAR100, STS-B, and SST-2 over the deployed pretraining model in Fig.~\ref{fig:training}. When $\varepsilon$ = 5, their accuracy is equal to 89.6\%, 71.8\%, 83.4\%, and 92.0\%, which is comparable to their non-private baseline training from scratch~\cite{WRN} and pretrained model~\cite{BERT,Roberta}, and better than state-of-the-art DP-SGD training from scratch \cite{DPSGD21,DPSGD22,DPSGD23}. This suggests pretraining technique is key to high utility CiM-based training with better tolerance against noise accumulation.  

\begin{figure}
    \centering
{\includegraphics[width=1.0\columnwidth]{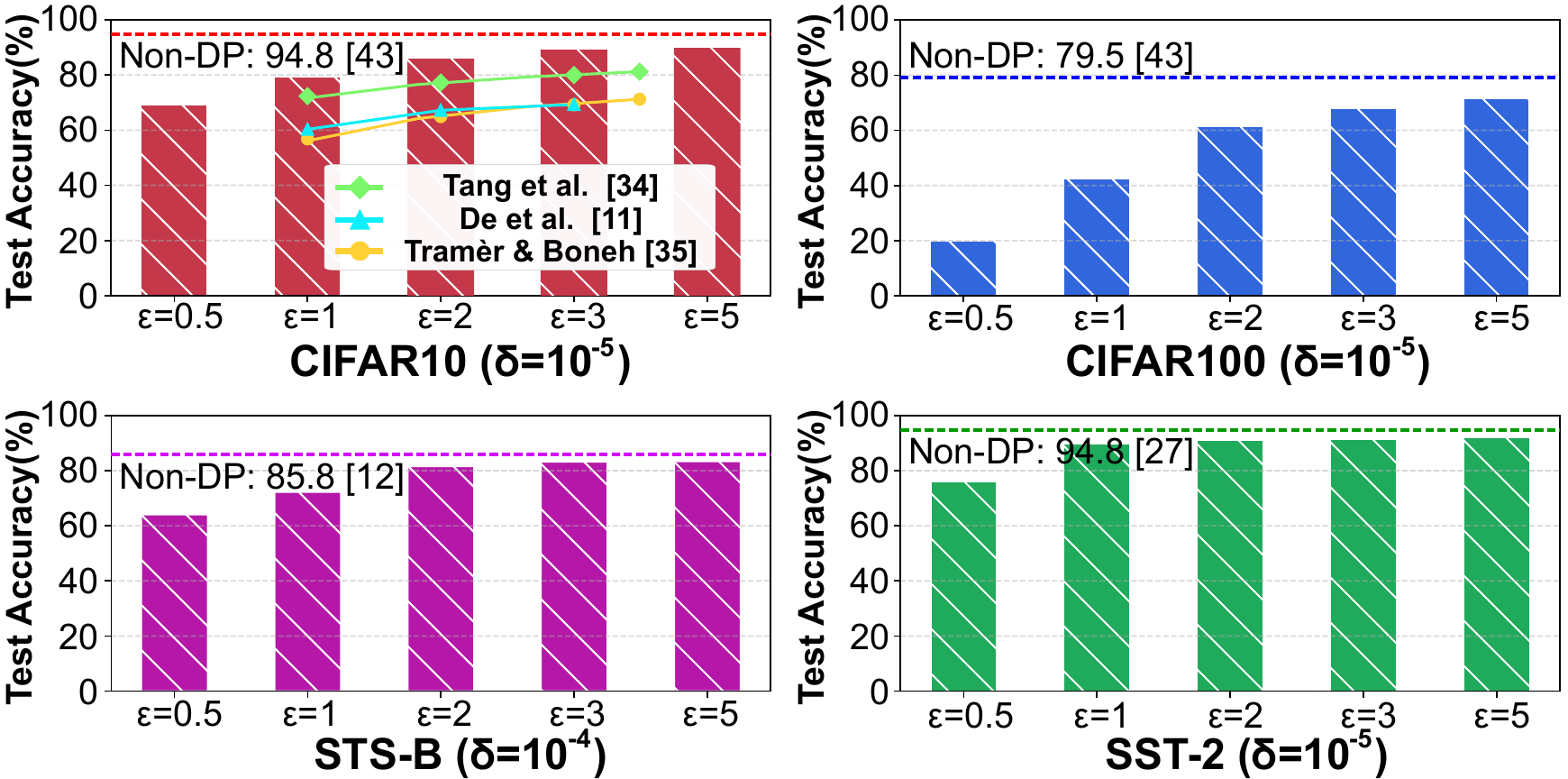}} \\
    \vspace{-7pt}  
    \caption{Experimental results for $\rramdp$-SGD with pretraining techniques ($\varepsilon$=2) as compared to non-private training (Non-DP) and DP-SGD from scratch.}
    \label{fig:training}%
    \vspace{-15pt}  
\end{figure}

\vspace{-5pt}  
\subsection{High-Performance RRAM-DP-SGD}
\textbf{Energy Saving.}
In Fig.~\ref{fig:energy}, energy savings are estimated by normalizing A100 energy consumption under theoretical limit in sec. 4.1 setting. The system operates at a 500MHz clock frequency. For in-memory computing, 512 RRAM crossbar tiles (32$\times$32 weight array in 3-device config.) execute 512$\times$32$\times$32 parallel MAC operations, with energy per MAC estimated at 0.447pJ ($\approx$FLOPS/2), including peripheral contributions (ADC: 72.3fJ, DAC: 46.3fJ, op-amp: 0.51pJ, RRAM driver: 0.27pJ)\cite{adc,dac,cai2020power}. RRAM inference uses a 0.2V read voltage with 50$\text{\textmu S}$ average conductance, while programming employs SET/RESET voltages of 1.2V/2.2V, with row-parallel programming 512$\times$32 weights executing at 29.2pJ per device per cycle. For the $\rramdp$-SGD accelerator at ($\varepsilon$=2, $\delta$=O(1/n))-DP, we achieve 54$\times$-57$\times$ energy reduction against A100 and 3$\times$-3.2$\times$ against DiVa-GEMM~\cite{diva}. Despite RRAM write overheads, the main energy cost (94\%-99\%) arises from per-sample gradient batch computations, with smaller batches benefiting less from energy savings.

\vspace{2pt}
\textbf{Time Speedup.}
In Fig.~\ref{fig:time}, speedup is similarly estimated by normalizing A100 computation time under theoretical limit in sec. 4.1 setting. The system leverages massive MAC parallelism through in-memory computing with 512$\times$32$\times$32 weight operations per MAC step and row-parallel programming supporting 512$\times$32 weights in one operation. Each programming cycle takes 40 ns, incorporating the time for SET/RESET and readout operations under the specified programming scheme. Each MAC operation takes 20 ns, with the overall execution highly optimized by the 500 MHz clock and parallel tile architecture.
We achieve 2.5$\times$-2.7$\times$ speedup compared to A100 and 1.7$\times$-1.8$\times$ compared to DiVa-GEMM~\cite{diva} at ($\varepsilon$=2, $\delta$=O(1/n))-DP, with major time consumed in MAC (94\%-99\%) compared to RRAM program time. This demonstrates the promising in-memory computing paradigm with privacy guarantees.

\textbf{Discussion on SOTA privacy accelerators.}
DiVa~\cite{diva} is the first architecture design dedicated to DP training, replacing the conventional weight-stationary systolic GEMM array with an outer-product dataflow engine to better accommodate the irregular, small inner-dimension matrix multiplications associated with per-sample gradient computation. However, DiVa remains a fully digital implementation with conventional memory hierarchies for MAC operations. Although
cryptographically safe pseudo-random number generator (CSPRNG) addresses floating-point error for digital accelerators, it renders additional latency and energy overhead (see secure mode in Opacus~\cite{Opacus} and Jax-privacy~\cite{jax}).
\begin{figure}[!t]
    \centering
    {\includegraphics[width=0.99\columnwidth]{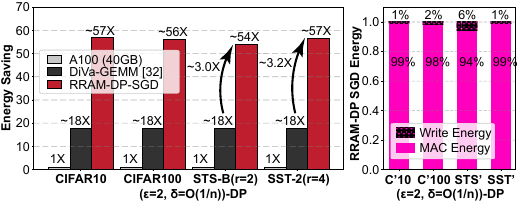}} \\
    \vspace{-10pt}  
    \caption{(LEFT) Experimental results for hardware energy savings as compared to GPU (A100) and DiVa-GEMM. (RIGHT) Energy consumption for RRAM. }
    \label{fig:energy}%
    \vspace{-10pt}  
\end{figure}

\begin{figure}[!t]
    \centering
    {\includegraphics[width=0.99\columnwidth]{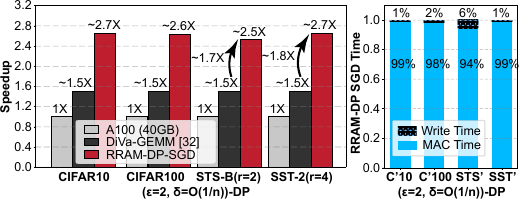}} \\
    \vspace{-10pt}  
    \caption{(LEFT) Experimental results for hardware speedup as compared to GPU (A100) and DiVa-GEMM. (RIGHT) Time consumption for RRAM. }
    \label{fig:time}%
    \vspace{-15pt}  
\end{figure}

\vspace{-5pt}
\section{Conclusion}
Former works assume idealized noise generation in hardware for privacy applications, which is not always true in real device statistics. In this work, we propose $\rramdp$ co-design to connect real RRAM device randomness for private on-device data storage and in-memory SGD with formal $(\varepsilon, \delta)$-DP privacy definition, with theoretical proof under $\mu$-GDP analysis framework. Our design identifies a robust RRAM noise sampling method with a single control variable, validated through empirical device statistics to inject sufficient $\rramdp$ mechanism noise for privacy protection. Additionally, we show that the pretraining technique inspired by DP-SGD can mitigate CiM-induced noisy learning, enabling near-lossless accurate CiM-training as compared to normal training. Our $\rramdp$-SGD accelerator establishes a novel computing paradigm for both efficient and private AIoT edge learning.

\vspace{-5pt}
\section{Acknowledgement}
This research was conducted by ACCESS – AI Chip Center for Emerging Smart Systems, supported by the InnoHK initiative of the Innovation and Technology Commission of the Hong Kong Special Administrative Region Government; supported by Hong Kong Research Grant Council - General Research Fund Scheme (Grant No. 17202422, 17212923, 17215025) Theme-based Research (Grant No.T45-701/22-R), and Strategic Topics Grant (Grant No.STG3/E-605/25-N).

\bibliographystyle{ACM-Reference-Format}
\bibliography{references}

@article{privacybook,
    author = {Dwork, Cynthia and Roth, Aaron},
    title = {The Algorithmic Foundations of Differential Privacy},
    year = {2014},
    journal = {Foundations and Trends in Theoretical Computer Science},
}

@article{Neyman_Pearson,
    author = {E.L. Lehmann, Joseph P. Romano},
    title = {Testing Statistical Hypotheses}, 
    journal = {Springer Texts in Statistics}, 
    year = {2022}, 
}

@inproceedings{sensitivity,
    author = {Dwork, Cynthia and others},
    title = {Calibrating Noise to Sensitivity in Private Data Analysis},
    year = {2006},
    booktitle = {Proceedings of Conference on Theory of Cryptography},
    series = {TCC '06}
}

@article{gaussDP,
    author = {Dong, Jinshuo and others},
    title = {Gaussian Differential Privacy},
    journal = {Journal of the Royal Statistical Society Series B: Statistical Methodology},
    year = {2022},
}

@article{gaussDP_Deep,
	author = {Bu, Zhiqi and others},
	journal = {Harvard Data Science Review},
	year = {2020},
	title = {Deep {Learning} {With} {Gaussian} {Differential} {Privacy}},
}

@inproceedings{dpsgd,
    title = {Deep Learning with Differential Privacy},
    author = {Abadi, Martin and others},
    booktitle = {Proceedings of Conference on Computer and Communications Security},
    year = {2016},
    series = {CCS '16},
}

@INPROCEEDINGS{linkage_attack01,
  author={Narayanan, Arvind and Shmatikov, Vitaly},
  booktitle={Proceedings of the IEEE Symposium on Security and Privacy}, 
  title={Robust De-anonymization of Large Sparse Datasets}, 
  year={2008},
  series={SP '08}
}

@inproceedings{floatingpt_attack01,
    author = {Mironov, Ilya},
    title = {On Significance of the Least Significant Bits for Differential Privacy},
    year = {2012},
    booktitle = {Proceedings of the ACM Conference on Computer and Communications Security}, 
    series = {CCS '12}
}

@inproceedings{floatingpt_attack02,
    author = {Jin, Jiankai and others},
    title = {Are We There Yet? Timing and Floating-Point Attacks on Differential Privacy Systems},
    year = {2022},
    booktitle = {Proceedings of the IEEE Symposium on Security and Privacy},
    series={SP '22}
}

@INPROCEEDINGS{MIA01,
    author={Carlini, Nicholas and others},
    booktitle={Proceedings of the IEEE Symposium on Security and Privacy}, 
    title={Membership Inference Attacks From First Principles}, 
    year={2022},
    series={SP '22}
}

@inproceedings{Model_inversion,
    author = {Fredrikson, Matt and others},
    title = {Model Inversion Attacks that Exploit Confidence Information and Basic Countermeasures},
    year = {2015},
    booktitle = {Proceedings of the ACM SIGSAC Conference on Computer and Communications Security},
    series={CCS '15}
}

@inproceedings{MIA_database,
    author = {Anderson, Maya and others},
    title = {Is My Data in Your Retrieval Database? Membership Inference Attacks Against Retrieval Augmented Generation},
    year = {2025},
    booktitle = {Proceedings of the International Conference on Information Systems Security and Privacy},
    series = {ICISSP '25}
}

@inproceedings{data_extract_LLM,
    author = {Carlini, Nicholas and others},
    title = {Extracting Training Data from Large Language Models},
    booktitle = {Proceedings of the USENIX Conference on Security Symposium},
    year = {2021},
    series = {SEC '21}
}

@inproceedings{data_extract_Diffusion,
    author = {Carlini, Nicholas and others},
    title = {Extracting Training Data from Diffusion Models},
    year = {2023},
    booktitle = {Proceedings of the USENIX Conference on Security Symposium},
    series = {SEC '23}
}

@article{AIoT,
    title = {{AIoT}-enabled smart surveillance for personal data digitalization: Contextual personalization-privacy paradox in smart home},
    journal = {Information \& Management},
    year = {2023},
    author = {Fengjiao Zhang and Zhao Pan and Yaobin Lu},
}

@inproceedings{LDP01,
    author = {Wang, Ning and others},
    booktitle = {Proceedings of the IEEE International Conference on Data Engineering},
    title = {Collecting and Analyzing Multidimensional Data with Local Differential Privacy},
    year = {2019},
    series={ICDE '19}
    
}

@article{RRAM_01,
    author = {Wang, Zhongrui and others},
    journal = {Nature reviews materials},
    title = {Resistive switching materials for information processing},
    year = {2020},
}

@article{ai_accelerator02,
  author={Wang, Zhongrui and others},
  journal={Nature Machine Intelligence}, 
  title={In situ training of feed-forward and recurrent convolutional memristor networks}, 
  year = {2019},
}

@inproceedings{snngx,
    author = {Wong, Kwunhang and others},
    title = {SNNGX: Securing Spiking Neural Networks with Genetic XOR Encryption on RRAM-based Neuromorphic Accelerator},
    year = {2025},
    booktitle = {Proceedings of the IEEE/ACM International Conference on Computer-Aided Design},
    series = {ICCAD '24}
}

@ARTICLE{Memristor_DP,
  author={Fu, Jingyan and others},
  journal={IEEE Internet of Things Journal}, 
  title={Memristor-Based Variation-Enabled Differentially Private Learning Systems for Edge Computing in {IoT}}, 
  year={2021},
}

@inproceedings{DPSGD23,
    author = {Tang, Xinyu and others},
    title = {Differentially Private Image Classification by Learning Priors from Random Processes},
    year = {2023},
    booktitle = {Proceedings of the International Conference on Neural Information Processing Systems},
    series = {NIPS '23}
}

@inproceedings{DPSGD22,
    author = {De, Soham and others},
    title = {Unlocking High-Accuracy Differentially Private Image Classification through Scale},
    year = {2022},
    booktitle = {Proceedings of the ICML Workshop on Theory and Practice of Differential Privacy},
}

@inproceedings{DPSGD21,
    author = {Tramèr, Florian and Boneh, Dan},
    title = {Differentially Private Learning Needs Better Features (or Much More Data)},
    booktitle = {Proceedings of the International Conference on Learning Representations},
    year = {2021},
    series = {ICLR '21}
}

@article{rram_noise,
  title={Resistive random access memory {(RRAM)}: An overview of materials, switching mechanism, performance, multilevel cell {(MLC)} storage, modeling, and applications},
  author={Zahoor, Furqan and Azni Zulkifli, Tun Zainal and Khanday, Farooq Ahmad},
  journal={Nanoscale research letters},
  year={2020},
}

@article{rram_program,
  title={Accurate Program/Verify Schemes of Resistive Switching Memory {(RRAM)} for In-Memory Neural Network Circuits},
  author={Milo, Valerio and others},
  journal={IEEE Transactions on Electron Devices},
  year={2021},
}

@article{ielmini2025resistive,
  title={Resistive switching random-access memory {(RRAM)}: Applications and requirements for memory and computing},
  author={Ielmini, Daniele and Pedretti, Giacomo},
  journal={Chemical Reviews},
  year={2025},
}

@misc{C10_dataset,
    title= {{CIFAR-10} and {CIFAR-100}},
    author= {Alex Krizhevsky and Vinod Nair and Geoffrey Hinton},
    howpublished={http://www.cs.toronto.edu/$\sim$kriz/cifar.html},
    year= {2009},
}

@inproceedings{glue,
    title = "{GLUE}: A Multi-Task Benchmark and Analysis Platform for Natural Language Understanding",
    author = "Wang, Alex  and others",
    booktitle = "Proceedings of the {EMNLP} Workshop on {B}lackbox{NLP}: Analyzing and Interpreting Neural Networks for {NLP}",
    year = "2018",
}

@article{Imagenet32,
      title={A Downsampled Variant of ImageNet as an Alternative to the CIFAR datasets}, 
      author={Patryk Chrabaszcz and Ilya Loshchilov and Frank Hutter},
      year={2017},
      journal={arXiv preprint 1707.08819},
}

@inproceedings{Opacus,
    title={Opacus: User-Friendly Differential Privacy Library in PyTorch},
    author={Ashkan Yousefpour and others},
    booktitle={Proceedings of the NeurIPS Workshop on Privacy in Machine Learning},
    year={2021},
}

@article{Jax,
      title={JAX-Privacy: A library for differentially private machine learning}, 
      author={Ryan McKenna and others},
      year={2026},
      journal={arXiv preprint 2602.17861},
}

@article{WRN,
      title={Wide Residual Networks}, 
      author={Sergey Zagoruyko and Nikos Komodakis},
      year={2017},
      journal={arXiv preprint 1605.07146},
}

@inproceedings{BERT,
    title = "{BERT}: Pre-training of Deep Bidirectional Transformers for Language Understanding",
    author = "Devlin, Jacob  and others",
    booktitle = "Proceedings of the Conference of the North {A}merican Chapter of the Association for Computational Linguistics: Human Language Technologies",
    year = "2019",
    series = {NAACL '19}
}

@article{Roberta,
    title={RoBERTa: A Robustly Optimized BERT Pretraining Approach},
    author={Yinhan Liu and others},
    year={2019},
    journal={arXiv preprint 1907.11692},
}

@inproceedings{diva,
  author={Park, Beomsik and others},
  booktitle={Proceedings of the IEEE/ACM International Symposium on Microarchitecture}, 
  title={DiVa: An Accelerator for Differentially Private Machine Learning}, 
  year={2022},
  series={MICRO '22}
}

@misc{calflops,
  author = {Ye, Xiaoju},
  title = {Calflops: a FLOPs and Params calculate tool for neural networks in pytorch framework},
  year = 2023,
  howpublished ={https://github.com/MrYxJ/calculate-flops.pytorch},
}

@article{LindCLT,
  author    = {Achim Klenke},
  title     = {Probability Theory: A Comprehensive Course},
  year      = {2014},
  journal = {Universitext}, 
}

@inproceedings{adc,
  title={A 67.51 dB SNDR 137.8 $\mu$W SAR-Assisted Slope ADC with High-Linearity Analog Slope and Low Switching Energy},
  author={Chu, Xiaofeng and Liang, Can and Cai, Zeyu},
  booktitle={Proceedings of the IEEE International Symposium on Circuits and Systems},
  year={2025},
  series={ISCAS '25}
}

@inproceedings{dac,
  title={An energy efficient 7.59-ENOB 50 MS/s flash-SAR ADC in 65-nm CMOS},
  author={Lee, Sanghyun and Kim, Youngmin},
  booktitle={Proceedings of the IEEE International Midwest Symposium on Circuits and Systems},
  year={2023},
  series={MWSCAS '23}
}

@article{cai2020power,
  title={Power-efficient combinatorial optimization using intrinsic noise in memristor Hopfield neural networks},
  author={Cai, Fuxi and others},
  journal={Nature Electronics},
  year={2020},
}

@inproceedings{bookkeeping_clip,
    author = {Bu, Zhiqi and others},
    title = {Differentially Private Optimization on Large Model at Small Cost},
    year = {2023},
    booktitle = {Proceedings of the 40th International Conference on Machine Learning},
    series = {ICML '23}
}

@article{eapu,
    author = {Liu, Jinchang and others},
    title = {Error-aware probabilistic training for memristive neural networks},
    year = {2025},
    journal = {Nature Communications}
}

@inproceedings{pipeline,
    author = {Haoxiong Ren and others},
    title = {When Pipelined In-Memory Accelerators Meet Spiking Direct Feedback Alignment: A Co-Design for Neuromorphic Edge Computing},
    booktitle = {Proceedings of the IEEE/ACM International Conference on Computer-Aided Design},
    year = {2025},
    series = {ICCAD '25}
}

@InProceedings{dp_addremove,
  title = {Practical and Private (Deep) Learning Without Sampling or Shuffling},
  author = {Kairouz, Peter and others},
  booktitle = {Proceedings of the 38th International Conference on Machine Learning},
  year = {2021},
  series = {ICML '21}
}

@inproceedings{choose_epsi,
  title={Differential Privacy: An Economic Method for Choosing Epsilon},
  author={Justin Hsu and others},
  booktitle={Proceedings of the IEEE 27th Computer Security Foundations Symposium},
  year={2014},
  series = {CSF '14}
}

\end{document}